\documentclass[11pt]{article}

\usepackage[utf8]{inputenc}
\usepackage[T1]{fontenc}
\usepackage{lmodern}
\usepackage[margin=1in]{geometry}
\usepackage{amsmath,amssymb,amsthm}
\usepackage{booktabs}
\usepackage{microtype}
\usepackage[hidelinks]{hyperref}
\usepackage{cleveref}

\makeatletter
\renewcommand\paragraph{\@startsection{paragraph}{4}{\z@}%
  {3.25ex \@plus1ex \@minus.2ex}{-1em}%
  {\normalfont\normalsize\itshape}}
\makeatother

\title{Nested Clustered Optimization Is One End of a Schur Bridge,\\
and the Interior Is Sometimes Provably Better}
\author{Peter Cotton\thanks{Microprediction.
Email: \texttt{peter.cotton@microprediction.com}. A companion
verification script is \texttt{paper/verify\_schur\_nco\_bridge.py}
at \texttt{github.com/microprediction/schur}.}}
\date{September 17, 2026}

\newcommand{\ones}{\mathbf{1}}
\DeclareMathOperator{\diag}{diag}
\DeclareMathOperator{\Cov}{Cov}
\DeclareMathOperator{\Var}{Var}

\newtheorem{proposition}{Proposition}
\newtheorem{remark}{Remark}
\newtheorem{definition}{Definition}
\newtheorem{corollary}{Corollary}

\begin{document}
\maketitle

\begin{abstract}
Nested clustered optimization allocates within each cluster from the
cluster's own covariance block and then across the resulting cluster
portfolios. Block inversion says the unconstrained minimum-variance
portfolio has the same two-tier shape, with each block replaced by its Schur
complement against every other asset. Conditioning instead on one knot from each other cluster truncates the
conditioning set in the manner of a Vecchia approximation, and we give the rank-one model of cross-cluster dependence under which it is
exact. Damping the complement by $\gamma\in[0,1]$ then gives a bridge with
nested clustered optimization at $\gamma=0$ and the global optimum at
$\gamma=1$, with no linear solve larger than a cluster or the number of
clusters. Under estimation error the optimal $\gamma$ can be strictly
interior and full coupling can remain optimal, and we give the local
theorem at the minimum-variance end with exact examples of both, including
a symmetric family in which the optimum is a closed form.
\end{abstract}

\section{Introduction}

Two-stage cluster methods partition the assets, allocate inside each
cluster, and then allocate across clusters. In nested clustered
optimization, NCO \cite{lopezdeprado2020}, the inner step optimizes
each cluster on its own covariance block and the outer step optimizes across
the resulting cluster portfolios. The outer problem has one dimension per
cluster, which is the source of the method's numerical stability.

The inner step reads only the cluster's own block. The only channel by which
one cluster's weights respond to another cluster is its scalar budget.

Schur-complementary allocation \cite{cotton2024schur} removes the same
blind spot from hierarchical risk parity \cite{lopezdeprado2016}. At
each split of a bisection tree it replaces a block by its Schur complement
against the sibling block, damped by $\gamma\in[0,1]$, and carries a
companion vector alongside. The result is a bridge with hierarchical risk
parity at $\gamma=0$ and minimum variance at $\gamma=1$.

This note builds the corresponding bridge for a flat partition. The left
end is NCO and the right end is the unconstrained global minimum-variance
portfolio. The obstacle is cost, since the complement of a cluster against
everything else needs an inverse of nearly full size.

Spatial statistics meets the same obstacle by conditioning on a few chosen
indices rather than on everything \cite{vecchia1988}. The conditionals
there are Schur complements, so the device transfers directly.

We borrow its vocabulary. One member of each cluster is its \emph{knot}, as
in the low-rank spatial models of Banerjee et al.~\cite{banerjee2008} that
Katzfuss and Guinness~\cite{katzfuss2021} place inside the Vecchia framework. The other
clusters' knots are a cluster's \emph{conditioning set}. We state the model
under which conditioning on that set loses nothing.

\Cref{sec:where} asks where to sit on the bridge when the covariance is
estimated. Part of the analysis of the tree in \cite{cotton2026bridge}
transfers, and exact examples show that both an interior optimum and a
locally optimal full coupling occur. \Cref{sec:rank} treats singular
covariances, where the parameterization by $\gamma$ itself becomes
singular.

All exactness statements concern the unconstrained, fully invested,
long-short problem. With box or long-only constraints no block
decomposition of the optimum is exact, here or on the tree.

\section{Setup}

Let $\Sigma\succ0$ be the covariance of $n$ assets partitioned into
clusters $I_1,\dots,I_k$. Cluster $i$ has a knot $p_i\in I_i$ and
remaining members $J_i=I_i\setminus\{p_i\}$. Write $P=\{p_1,\dots,p_k\}$
and $P_{-i}=P\setminus\{p_i\}$. For index sets $X,Y$ write $\Sigma_{XY}$ for the
corresponding block, $\Sigma_{ii}$ for the block of cluster $i$, and $-i$
for the complement of $I_i$.

A two-stage method chooses a portfolio $v_i$ for each cluster, supported on
$I_i$ with $\ones^\top v_i=1$, and then budgets across the clusters. With
$V$ the $n\times k$ matrix of cluster portfolios,
\begin{equation}
  a \;=\; \operatorname{alloc}\!\big(V^\top\Sigma V\big),
  \qquad w \;=\; Va .
  \label{eq:twostage}
\end{equation}
NCO takes $v_i$ to be the minimum-variance portfolio of $\Sigma_{ii}$ and
minimum variance for the outer allocation.

\section{Block inversion}

Fix a vector $u$ and a cluster $i$. Block inversion of $\Sigma$ on the
partition $(I_i,-i)$ gives \cite[App.~A]{cotton2024schur}
\begin{equation}
  \big(\Sigma^{-1}u\big)_{I_i}
  \;=\; \big(\Sigma^{c}_{ii}\big)^{-1} b_i,
  \label{eq:blockinv}
\end{equation}
where
\begin{equation}
  \Sigma^{c}_{ii} = \Sigma_{ii}-\Sigma_{i,-i}\Sigma_{-i,-i}^{-1}\Sigma_{-i,i},
  \qquad
  b_i = u_{I_i}-\Sigma_{i,-i}\Sigma_{-i,-i}^{-1}u_{-i} .
  \label{eq:complement}
\end{equation}
For a positive-definite $Q$ and a vector $b$ we call $Q^{-1}b$ a
\emph{generalized minimum-variance direction}. Up to scale it minimizes
$x^\top Qx$ subject to $b^\top x=1$, which is ordinary minimum variance only
when $b=\ones$.

With $u=\ones$, the blocks of the global minimum-variance direction
$\Sigma^{-1}\ones$ are therefore generalized minimum-variance directions,
one per cluster, on the cluster's covariance conditional on everything
else. There is no cross-cluster optimizer in the global optimum. The
cluster totals are whatever the stacked directions sum to, and every piece
of cross-cluster information enters through the conditioning.

NCO has the same shape with the conditioning removed. Restoring it costs an
inverse of size $n-|I_i|$ per cluster, which is what a two-stage method
exists to avoid.

\section{The gateway model}

\begin{definition}[Gateway model]
For every cluster $j$, the residual of the remaining members $J_j$ after linear
regression on their own knot is uncorrelated with every asset outside
$I_j$. Equivalently the cross-cluster residual
\begin{equation}
  R_j(p_j) \;=\; \Sigma_{J_j,-j}-\Sigma_{J_j p_j}\,\sigma_{p_j}^{-2}\,\Sigma_{p_j,-j}
  \label{eq:residual}
\end{equation}
is zero.
\end{definition}

A member may be correlated with anything, but its correlation with other
clusters passes entirely through its knot. The knot is the cluster's only
gateway. Each cross-cluster block of $\Sigma$ then has rank at most one.

This is a strong assumption, and clustering alone does not imply it. It
asks that member residuals be uncorrelated both with the other knots and
with the other clusters' member residuals.

\begin{proposition}[Knots are sufficient]
\label{prop:sufficient}
Under the gateway model, for each cluster $i$ and every $u$,
\begin{equation}
  \Sigma^{c}_{ii}
  \;=\; \Sigma_{ii}-\Sigma_{i,P_{-i}}\,\Sigma_{P_{-i}P_{-i}}^{-1}\,\Sigma_{P_{-i},i},
  \qquad
  b_i \;=\; u_{I_i}-\Sigma_{i,P_{-i}}\,\Sigma_{P_{-i}P_{-i}}^{-1}\,u_{P_{-i}} .
  \label{eq:cheap}
\end{equation}
\end{proposition}

\begin{proof}
The matrix $\Sigma_{i,-i}\Sigma_{-i,-i}^{-1}$ holds the coefficients of the
linear regression of cluster $i$ on all other assets. Split the regressors
into the other knots $P_{-i}$ and the other remaining members $J_{-i}$, and write
$J_{-i}=\beta\,P_{-i}+\varepsilon$ for the block-diagonal regression of
each set of remaining members on its own knot. The map from $(P_{-i},J_{-i})$ to
$(P_{-i},\varepsilon)$ is invertible, so the regression may be run on the
latter pair. Under the gateway model $\varepsilon$ is uncorrelated with
cluster $i$ and with $P_{-i}$, so its coefficients vanish and the
coefficients on $P_{-i}$ are $\Sigma_{i,P_{-i}}\Sigma_{P_{-i}P_{-i}}^{-1}$.
Substituting into \eqref{eq:complement} gives \eqref{eq:cheap}, with
$u_{-i}$ contributing only through its $P_{-i}$ entries.
\end{proof}

Conditioning cluster $i$ on the other knots is the same as conditioning it
on everything. The inverse in \eqref{eq:cheap} is $(k-1)\times(k-1)$.

\begin{remark}[Relation to Vecchia approximations]
A Vecchia approximation orders the variables and replaces each conditioning
set by a small subset of the predecessors, which defines one acyclic
approximate joint law \cite{vecchia1988,katzfuss2021}. The conditionals
are Schur complements. Replacing $-i$ by $P_{-i}$ in \eqref{eq:complement}
is the same truncation of a conditioning set, but it is applied to the full
conditional of every cluster at once. Those conditionals are mutually
cyclic and no approximate joint law is formed, so the recipe is
Vecchia-inspired rather than a Vecchia approximation.
\end{remark}

The gateway model itself does have an exact Vecchia form.

\begin{corollary}[Gaussian factorization]
\label{cor:gaussian}
If the returns $r$ are jointly Gaussian, the gateway model holds if and only
if
\begin{equation}
  f(r) \;=\; f(r_P)\,\prod_{i=1}^{k} f\big(r_{J_i}\mid r_{p_i}\big).
  \label{eq:dag}
\end{equation}
\end{corollary}

\begin{proof}
Under the model the residuals $\varepsilon_i=r_{J_i}-\beta_ir_{p_i}$ are
uncorrelated with $r_P$ and with each other, hence independent under
Gaussianity, which gives \eqref{eq:dag}. Conversely \eqref{eq:dag} makes
$\varepsilon_i$ independent of every asset outside $I_i$.
\end{proof}

With the knots ordered first, \eqref{eq:dag} is a block-Vecchia
factorization, a sparser relative of the knot-based cases collected by
Katzfuss and Guinness~\cite{katzfuss2021}. The covariance results in this note do not need
Gaussianity.

\begin{remark}[Testing the model]
Regressing each remaining member on its own knot and the other knots, and
checking that the other-knot coefficients vanish, does not test the model.
Member residuals can be orthogonal to every knot and still be correlated
with each other across clusters, and the sufficiency of the knots then
fails. The model holds exactly when $R_j(p_j)=0$ for every $j$, which can
be checked from the sample covariance before any portfolio is built.
\end{remark}

\begin{remark}[Measuring the loss and choosing knots]
The raw size of $R_j$ does not measure what is lost when the model fails. A
residual with variance $\delta^2$ and covariance $c\delta$ with an outside
asset has $R=c\delta\to0$, while the Schur correction it carries, $c^2$,
does not move. The loss for cluster $i$ is the omitted correction
\begin{equation}
  \Delta_i \;=\; \Sigma_{i,-i}\Sigma_{-i,-i}^{-1}\Sigma_{-i,i}
   -\Sigma_{i,P_{-i}}\Sigma_{P_{-i}P_{-i}}^{-1}\Sigma_{P_{-i},i}
  \;\succeq\; 0,
  \label{eq:loss}
\end{equation}
the part of cluster $i$ explained by the other assets beyond their knots.
A scale-free screen for a candidate knot $p\in I_j$ is the largest canonical
correlation between the residual of $I_j\setminus\{p\}$ on $p$ and the
assets outside $I_j$. Both need a large inverse, but they are diagnostics
computed once and are not part of the allocation.
\end{remark}

\begin{remark}[A symmetric factor model]
Suppose instead that every member of cluster $j$ loads on a latent cluster
factor, $r_m=\lambda_m f_j+\eta_m$, with idiosyncratic terms uncorrelated
with each other and with the factors. Then
$R_j(p)=(1-\rho_p)\,\lambda_{J_j}\Cov(f_j,r_{-j})$, where
$\rho_p=\lambda_p^2\Var(f_j)/\sigma_p^2$ is the share of the knot's variance
explained by the factor. The gateway model holds to the extent the knot is
the factor. A knot with a high $\rho_p$ is then the natural choice, and the
member most aligned with the cluster's first principal component
approximates it.
\end{remark}

\section{The bridge}

Following \cite{cotton2024schur}, damp the conditioning. For
$\gamma\in[0,1]$ define the \emph{conditioned pair} of cluster $i$,
\begin{equation}
\begin{aligned}
  Q_i(\gamma) &= \Sigma_{ii}-\gamma\,\Sigma_{i,P_{-i}}\Sigma_{P_{-i}P_{-i}}^{-1}\Sigma_{P_{-i},i},\\
  b_i(\gamma) &= u_{I_i}-\gamma\,\Sigma_{i,P_{-i}}\Sigma_{P_{-i}P_{-i}}^{-1}u_{P_{-i}},
\end{aligned}
  \label{eq:pair}
\end{equation}
and the direction $d_i(\gamma)=Q_i(\gamma)^{-1}b_i(\gamma)$. Let $D_\gamma$
be the $n\times k$ matrix with $d_i(\gamma)$ in column $i$, supported on
$I_i$. The bridge portfolio is
\begin{equation}
  w_\gamma \;\propto\; D_\gamma\big(D_\gamma^\top\Sigma D_\gamma\big)^{-1}D_\gamma^\top u .
  \label{eq:top}
\end{equation}
It solves, on the span of the $k$ directions, the problem that
$\Sigma^{-1}u$ solves on the whole space. Since $b_i(\gamma)=u_{I_i}-\gamma a_i$
with $a_i=\Sigma_{i,P_{-i}}\Sigma_{P_{-i}P_{-i}}^{-1}u_{P_{-i}}$ fixed, a
direction can vanish only at isolated $\gamma$, and there its shape has
the limit $-Q_i(\gamma)^{-1}a_i$. We use that limiting direction, which
keeps the span continuous in $\gamma$; only its scale is lost, and scale
does not enter \eqref{eq:top}. A cluster is dropped only when $u_{I_i}=0$.

\paragraph{Two-tier reading.} The right side of \eqref{eq:top} does not
change when a column of $D_\gamma$ is rescaled. When every total
$\ones^\top d_i(\gamma)$ is nonzero, set $v_i=d_i/\ones^\top d_i$. Then
\eqref{eq:top} is the two-stage method \eqref{eq:twostage} with
$a\propto(V^\top\Sigma V)^{-1}V^\top u$. The induced outer inputs are the
$k\times k$ covariance $V^\top\Sigma V$ and the $k$-vector $V^\top u$, which
is $\ones$ when $u=\ones$. The unnormalized form \eqref{eq:top} has no
singularity when a cluster total crosses zero.

\begin{proposition}[The two ends of the bridge]
\label{prop:bridge}
\leavevmode
\begin{enumerate}
\item At $\gamma=0$ with $u=\ones$ the recipe is NCO with minimum variance
  at both tiers.
\item Under the gateway model the nonvanishing columns of $D_1$ are the
  blocks of $\Sigma^{-1}u$, and $w_1\propto\Sigma^{-1}u$.
\end{enumerate}
\end{proposition}

\begin{proof}
At $\gamma=0$ the pair is $(\Sigma_{ii},\ones)$, so $d_i=\Sigma_{ii}^{-1}\ones$
has a positive total, $v_i$ is the minimum-variance portfolio of the block,
and the outer tier is minimum variance across the cluster portfolios. For
the second claim, \eqref{eq:blockinv} and \cref{prop:sufficient} give
$\Sigma^{-1}u$ as the sum of the nonvanishing columns of $D_1$, which lies
in the column span of $D_1$ whatever replaces a vanishing column. Then
$D_1(D_1^\top\Sigma D_1)^{-1}D_1^\top u
=D_1(D_1^\top\Sigma D_1)^{-1}D_1^\top\Sigma D_1\alpha=D_1\alpha$ for the
$\alpha$ with $\Sigma^{-1}u=D_1\alpha$.
\end{proof}

\begin{table}[ht]
\centering
\begin{tabular}{@{}ll@{}}
\toprule
$\gamma$ & Method \\
\midrule
$0$ & NCO \cite{lopezdeprado2020} \\
$(0,1)$ & the interior of the bridge \\
$1$ & unconstrained global minimum variance, under the gateway model \\
\bottomrule
\end{tabular}
\caption{The bridge with $u=\ones$ and minimum variance at both tiers.}
\label{tab:bridge}
\end{table}

\paragraph{Scope.} The right end is the unconstrained, fully invested,
long-short solution. No condition on cluster totals is needed for
\eqref{eq:top}. A global solution can hold a self-financing position inside
a cluster with zero net weight. The two-tier reading fails there, because
no cluster budget times a normalized cluster portfolio represents it, and
\eqref{eq:top} does not.

\begin{remark}[Other objectives]
\Cref{prop:bridge} is stated for a direction $\Sigma^{-1}u$. With $u=\mu$
the direction $\Sigma^{-1}\mu$ maximizes the Sharpe ratio. Its fully
invested multiple does so only when $\ones^\top\Sigma^{-1}\mu>0$. When that
total is negative the fully invested multiple has the opposite orientation
and the lowest signed Sharpe ratio.

Mean--variance problems with a budget
constraint have solutions that combine $\Sigma^{-1}\mu$ and
$\Sigma^{-1}\ones$, and each part is covered separately. Risk parity is not
of this form. The recipe still runs with a risk-parity outer tier, but
nothing is exact at $\gamma=1$.
\end{remark}

\begin{remark}[Rules that read only a covariance]
An inner rule that cannot take a constraint vector can be applied after a
change of variables. With $y=b_i\circ x$ the constraint $b_i^\top x=1$
becomes $\ones^\top y=1$ and the covariance becomes
$M=\diag(b_i)^{-1}Q_i\diag(b_i)^{-1}$. Apply the rule to $M$ and map back by
$x=y/b_i$, coordinate by coordinate. For minimum variance this returns
$Q_i^{-1}b_i$. It requires $b_i$ to have no zero entry. Applying minimum
variance to $M$ without mapping back returns $\diag(b_i)Q_i^{-1}b_i$, which
is a different portfolio.
\end{remark}

\begin{remark}[Relation to the tree]
Schur complements compose: the complement of a complement is the complement
against the union. So at $\gamma=1$ the recursion of
\cite{cotton2024schur}, run down any tree whose leaves are the clusters
and stopped there, hands each cluster the pair \eqref{eq:complement}. The
flat construction and the iterated one agree at the right end. They differ
in the interior, because the tree damps at every split and the damping
compounds along the path, while \eqref{eq:pair} damps each cluster once
against all the other knots.
\end{remark}

\section{Structure under the gateway model}
\label{sec:structure}

The gateway model gives both the outer covariance and the cluster
directions a simple form.

\begin{proposition}[Knot covariance plus a ridge]
\label{prop:outer}
Under the gateway model, let $v_i$ be any portfolios supported on their
clusters. Write $\beta_i=\Sigma_{J_ip_i}/\sigma_{p_i}^2$ for the member
betas, $E_i=\Sigma_{J_iJ_i}-\beta_i\beta_i^\top\sigma_{p_i}^2$ for the
residual covariance, and $\kappa_i=v_{i,p_i}+\beta_i^\top v_{i,J_i}$ for
the cluster portfolio's beta to its knot. Then
\begin{equation}
  V^\top\Sigma V \;=\; K\,\Sigma_{PP}\,K + D,
  \qquad K=\diag(\kappa_i),
  \quad D=\diag\!\big(v_{i,J_i}^\top E_i\,v_{i,J_i}\big).
  \label{eq:outer}
\end{equation}
\end{proposition}

\begin{proof}
Under the gateway model $\Sigma_{J_i,l}=\beta_i\Sigma_{p_i,l}$ for every
asset $l$ outside $I_i$. Applying this on both sides gives
$v_i^\top\Sigma_{I_iI_l}v_l=\kappa_i\kappa_l\Sigma_{p_ip_l}$ for $i\neq l$.
On the diagonal, the knot is uncorrelated with the member residual, so
$v_i^\top\Sigma_{ii}v_i=\kappa_i^2\sigma_{p_i}^2+v_{i,J_i}^\top E_iv_{i,J_i}$.
\end{proof}

The outer problem sees the knot covariance, scaled by the cluster
portfolios' knot betas, plus their residual variances on the diagonal. The
diagonal term raises the smallest eigenvalue of the outer covariance by at
least the smallest residual variance. The remaining members supply a ridge, which
is one reading of why the outer tier of NCO is stable.

\begin{proposition}[Only the knots move]
\label{prop:knotsmove}
Under the gateway model, split the direction $d_i(\gamma)$ into its member
part $d_{i,J_i}$ and its knot exposure
$\kappa_i=d_{i,p_i}+\beta_i^\top d_{i,J_i}$. Then
\begin{equation}
  d_{i,J_i}(\gamma)=E_i^{-1}\big(u_{J_i}-\beta_i\,u_{p_i}\big),
  \qquad
  \kappa_i(\gamma)=\frac{u_{p_i}-\gamma\,a_i}{\sigma_{p_i}^2-\gamma\,c_i},
  \label{eq:knotsmove}
\end{equation}
where $a_i=\Sigma_{p_i,P_{-i}}\Sigma_{P_{-i}P_{-i}}^{-1}u_{P_{-i}}$ and
$c_i=\Sigma_{p_i,P_{-i}}\Sigma_{P_{-i}P_{-i}}^{-1}\Sigma_{P_{-i},p_i}$. The
member part depends on neither $\gamma$ nor any other cluster.
\end{proposition}

\begin{proof}
Under the model $\Sigma_{I_i,P_{-i}}=\binom{1}{\beta_i}\Sigma_{p_i,P_{-i}}$,
so $Q_i(\gamma)=\Sigma_{ii}-\gamma c_i\binom{1}{\beta_i}\binom{1}{\beta_i}^{\!\top}$
and $b_i(\gamma)=u_{I_i}-\gamma a_i\binom{1}{\beta_i}$. In the coordinates
$(\kappa,x_J)$ the quadratic form $x^\top Q_ix$ is
$(\sigma_{p_i}^2-\gamma c_i)\kappa^2+x_J^\top E_ix_J$ and the constraint
$b_i^\top x$ is $(u_{p_i}-\gamma a_i)\kappa+(u_{J_i}-\beta_iu_{p_i})^\top x_J$.
The two parts decouple, and $Q_i^{-1}b_i$ is \eqref{eq:knotsmove}.
\end{proof}

NCO and the global optimum hold the same member position inside every
cluster, up to the cluster's scale. Moving along the bridge reprices only
the knot. With $u=\ones$ a member earns weight in proportion to $1-\beta$,
its shortfall in beta to the knot, per unit of residual variance.

\begin{corollary}[A cluster is a knot and one independent asset]
\label{cor:compress}
Take $u=\ones$ and put
$\delta_i=(\ones-\beta_i)^\top E_i^{-1}(\ones-\beta_i)$. Under the gateway
model, replace each cluster by its knot together with one asset that is
uncorrelated with everything and has variance $1/\delta_i$. Then
$V_0(\gamma)$ is unchanged for every $\gamma$. So is $F(\gamma,\tau)$ of
\cref{sec:where} under noise that perturbs only the knot covariance and
preserves the gateway structure.
\end{corollary}

\begin{proof}
By \cref{prop:knotsmove} the direction of cluster $i$ has knot exposure
$\kappa_i(\gamma)$, total weight $\kappa_i(\gamma)+\delta_i$ and residual
variance $\delta_i$. By \cref{prop:outer} the outer tier, and the variance
of the resulting portfolio, read the cluster only through these and the
knot covariance. The replacement cluster has $\beta=0$ and $E=1/\delta_i$,
which gives the same three quantities.
\end{proof}

A large cluster interior therefore never obscures the analysis: it enters
as one scalar. If every member has unit beta to its knot then
$\delta_i=0$, the members carry no budget, and $w_\gamma$ is the same for
every $\gamma$ at which the knot exposures are nonzero.

For the rest of this section take $u=\ones$ and write $S=\Sigma_{PP}$ for
the knot covariance, $s_i=S_{ii}$, $h=S^{-1}\ones$ and
$\nu_i=1/(S^{-1})_{ii}$, the variance of knot $i$ left unexplained by the
other knots. In the notation of \cref{prop:knotsmove},
$\nu_i=s_i-c_i$ and $h_i=(1-a_i)/\nu_i$.

\begin{proposition}[Effective damping]
\label{prop:effective}
With $r_i=\nu_i/s_i$,
\begin{equation}
  \kappa_i(\gamma)=(1-\lambda_i)\,\frac{1}{s_i}+\lambda_i\,h_i,
  \qquad
  \lambda_i=\frac{\gamma\,r_i}{1-\gamma+\gamma\,r_i},
  \qquad\text{that is}\qquad
  \frac{\lambda_i}{1-\lambda_i}=r_i\,\frac{\gamma}{1-\gamma}.
  \label{eq:effective}
\end{equation}
\end{proposition}

\begin{proof}
Write the denominator of $\kappa_i$ as $s_i(1-\gamma)+\gamma\nu_i$ and the
numerator as $(1-\gamma)+\gamma\nu_ih_i$, and divide.
\end{proof}

The knot exposure interpolates between its NCO value $1/s_i$ and its
minimum-variance value $h_i$, and $\lambda_i$ is the interpolation actually
achieved. A common $\gamma$ does not damp the clusters equally. A knot that
the other knots nearly explain has small $r_i$, and its exposure barely
moves until $\gamma$ is very close to $1$.

\begin{proposition}[Lost precision]
\label{prop:lost}
Suppose every $\delta_i>0$. Put $\Delta=\operatorname{diag}(\delta_i)$,
$K_\gamma=\operatorname{diag}(\kappa_i(\gamma))$,
$Z=\ones^\top h+\sum_i\delta_i$ and
$W_\gamma=\big(S^{-1}+K_\gamma\Delta^{-1}K_\gamma\big)^{-1}$. The minimum
variance is $V_\star=1/Z$, and
\begin{equation}
  V_0(\gamma)-V_\star=\frac{L_\gamma}{Z\,(Z-L_\gamma)},
  \qquad
  L_\gamma=(h-\kappa)^\top W_\gamma\,(h-\kappa).
  \label{eq:lost}
\end{equation}
\end{proposition}

\begin{proof}
By \cref{cor:compress} work with $2k$ assets of covariance
$\operatorname{diag}(S,\Delta^{-1})$ and budget vector $(\ones,\ones)$. The
unconstrained precision is $Z$ and its direction is $(h,\delta)$. The bridge
restricts the knot exposures to $K_\gamma a$ and the member weights to
$\Delta a$, and its precision $1/V_0(\gamma)$ falls short of $Z$ by the
squared distance from $(h,\delta)$ to that subspace,
$\min_a\{(h-K_\gamma a)^\top S\,(h-K_\gamma a)
+(\ones-a)^\top\Delta\,(\ones-a)\}$. Eliminating $a$ and applying the
Woodbury identity gives $L_\gamma$.
\end{proof}

Since $W_\gamma\succ0$, the bridge attains the minimum variance exactly when
$\kappa=h$, and
\begin{equation}
  \kappa_i(\gamma)-h_i=\frac{(1-\gamma)\,(1-s_ih_i)}{(1-\gamma)s_i+\gamma\nu_i}.
  \label{eq:kminush}
\end{equation}
So if some $s_ih_i\neq1$ then $\gamma=1$ is the unique global minimizer of
$V_0$ on $[0,1]$. If every $s_ih_i=1$ the whole bridge is already optimal.
The curvature at the end is explicit as well. With
$p_i=\kappa_i'(1)=(s_ih_i-1)/\nu_i$ and
$W_1=\big(S^{-1}+\operatorname{diag}(h_i^2/\delta_i)\big)^{-1}$,
\begin{equation}
  V_0''(1)=\frac{2\,p^\top W_1\,p}{Z^2}.
  \label{eq:curvature}
\end{equation}

\paragraph{Equicorrelated knots.} Let the $k$ knots have unit variance and
common correlation $\rho$, and take $u=\ones$. Then
$a=\rho(k-1)/\big(1+(k-2)\rho\big)$ and $c=\rho\,a$, so every cluster has
\begin{equation*}
  \kappa(\gamma)=\frac{1+(k-2)\rho-\gamma\,\rho\,(k-1)}{1+(k-2)\rho-\gamma\,\rho^2(k-1)} .
\end{equation*}
For $\rho>0$ the knot exposure falls from $1$ at the NCO end to
$1/\big(1+(k-1)\rho\big)$ at the minimum-variance end, which is the scale of
the equicorrelated minimum-variance weight. The more knots and the higher
their correlation, the further the bridge has to travel.

\begin{remark}[Degenerate partitions]
With every asset its own cluster each direction is a multiple of a
coordinate vector, so with the limiting direction at a vanishing the
directions span the whole space at every $\gamma$. With a single cluster
there is nothing to condition on. In both cases $w_\gamma$ is the minimum-variance portfolio for
every $\gamma$, and no model is needed. The bridge has length only for
partitions strictly between the two.
\end{remark}

\section{Where to sit on the bridge}
\label{sec:where}

The cross-knot regression coefficients in \eqref{eq:pair} are estimated, and
$\gamma$ is how much of that estimate to trust. The analysis of the tree in
\cite{cotton2026bridge} transfers in part. This section takes $u=\ones$.

Let $\widehat\Sigma_\tau=\Sigma+\tau E$ be an estimate, where $E$ is a random
symmetric matrix whose law is unchanged by $E\mapsto-E$, and
$\widehat\Sigma_\tau\succ0$ on the range of $\tau$ considered. Every
quantity in the recipe is computed from $\widehat\Sigma_\tau$. The partition
and the knots are held fixed as $\tau$ varies, since reclustering or
reselecting knots from the data can destroy smoothness. The portfolio
$w_\gamma$ of \eqref{eq:top} is normalized by $\ones^\top w_\gamma=1$, and
we assume the normalizing total stays nonzero. Define the expected
out-of-sample variance and its noiseless profile,
\begin{equation}
  F(\gamma,\tau)=\mathbb{E}\big[w_\gamma(\widehat\Sigma_\tau)^\top\,\Sigma\,
  w_\gamma(\widehat\Sigma_\tau)\big],
  \qquad V_0(\gamma)=F(\gamma,0).
  \label{eq:F}
\end{equation}
It is even in $\tau$. We assume $F$ is smooth, put
$G(\gamma)=\tfrac12\,\partial_\tau^2F(\gamma,0)$, and assume the expansion
$F=V_0+\tau^2G+O(\tau^4)$ holds in $C^2$ as a function of $\gamma$.

\begin{proposition}[Incremental noise cost]
\label{prop:incremental}
Let $\Phi(\gamma,\tau)=F(\gamma,\tau)-F(0,\tau)-\big[V_0(\gamma)-V_0(0)\big]$.
Then $\Phi(\gamma,\tau)=\gamma\,\tau^2H(\gamma,\tau)$ for a smooth $H$.
\end{proposition}

\begin{proof}
$\Phi(0,\tau)=0$, so $\Phi=\gamma\,\Phi_1$ with $\Phi_1$ smooth. Also
$\Phi(\gamma,0)=0$, so $\Phi_1(\gamma,0)=0$, and $\Phi_1$ is even in $\tau$,
so $\Phi_1=\tau^2H$.
\end{proof}

On the tree, under noise supported only on the top-level cross-block, the
hierarchical end does not read that block, and the uncentred cost $F-V_0$
already carries the factor $\gamma$ \cite{cotton2026bridge}. Under
whole-matrix noise that fails on the tree as well, because $F(0,\tau)$ is
then noisy.

NCO is different even under cross-cluster noise alone. Its outer tier reads
cross-cluster covariance at $\gamma=0$, so NCO has an estimation cost of its
own, and the factor $\gamma$ appears only after centring at the noisy NCO
end. The centred statement holds for the general noise of this section, and
the cost of moving away from NCO is still of order $\gamma\tau^2$.

\begin{corollary}[Separation from NCO]
\label{cor:separation}
If $V_0'(0)<0$ then $\partial_\gamma F(0,\tau)<0$ for all sufficiently small
$\tau$, and NCO is strictly suboptimal.
\end{corollary}

\begin{proof}
By \cref{prop:incremental},
$\partial_\gamma F(0,\tau)=V_0'(0)+\tau^2H(0,\tau)$.
\end{proof}

\begin{proposition}[The minimum-variance end]
\label{prop:mvend}
Under the gateway model $V_0'(1)=0$ and
$V_0''(1)=2\,\dot w^\top\Sigma\,\dot w$, where
$\dot w=\partial_\gamma w_\gamma$ at $\gamma=1$. Suppose $\dot w\neq0$. Then
for all sufficiently small $\tau$ the map $F(\cdot,\tau)$ has a unique
stationary point near $1$, a strict local minimizer, at
\begin{equation}
  \tilde\gamma(\tau)\;=\;1-\frac{G'(1)}{V_0''(1)}\,\tau^2+O(\tau^4).
  \label{eq:shift}
\end{equation}
If $G'(1)>0$ this local minimizer is strictly inside the bridge for small
$\tau>0$. If $G'(1)<0$ then $\gamma=1$ is a strict local minimizer of
$F(\cdot,\tau)$ on $[0,1]$. If $G'(1)=0$ the second order is silent and the
fourth or a higher order decides.

Suppose in addition that $\gamma=1$ is the unique global minimizer of $V_0$
on $[0,1]$ and that $F(\cdot,\tau)\to V_0$ uniformly on $[0,1]$. Then for
all sufficiently small $\tau$ these local minimizers are the global optimum
over $[0,1]$.
\end{proposition}

\begin{proof}
Since $\ones^\top w_\gamma=1$ we have $\ones^\top\dot w=0$ and
$\ones^\top\ddot w=0$. By \cref{prop:bridge},
$w_1=\Sigma^{-1}\ones/\ones^\top\Sigma^{-1}\ones$, so
$\Sigma w_1\propto\ones$. Hence $V_0'(1)=2\dot w^\top\Sigma w_1=0$, and in
$V_0''(1)=2\dot w^\top\Sigma\dot w+2\ddot w^\top\Sigma w_1$ the second term
vanishes. Apply the implicit function theorem to
$\partial_\gamma F(\gamma,\tau)=0$ at $(1,0)$, where
$\partial_\gamma^2F=V_0''(1)>0$. This gives \eqref{eq:shift} and strict
convexity of $F(\cdot,\tau)$ on a neighbourhood of $1$ that does not depend
on $\tau$.

When $G'(1)<0$ the stationary point lies beyond $1$, so
$F(\cdot,\tau)$ decreases up to the boundary. Under the additional
hypotheses the global minimizers of $F(\cdot,\tau)$ converge to $1$, so for
small $\tau$ they lie in that neighbourhood.
\end{proof}

When every $\delta_i>0$ the uniqueness hypothesis is supplied by
\cref{prop:lost}, unless the whole bridge is already optimal. It remains a
hypothesis when some $\delta_i=0$, since gateway exactness alone does not
exclude another $\gamma<1$ that produces the same portfolio.

The displacement in \eqref{eq:shift} is of second order in the noise, and
the gain from making it is of fourth order. When $G'(1)>0$,
\begin{equation}
  F(1,\tau)-F(\tilde\gamma(\tau),\tau)
  =\frac{G'(1)^2}{2\,V_0''(1)}\,\tau^4+O(\tau^6).
  \label{eq:gain}
\end{equation}
A large gain at finite noise should not be inferred from the local sign.

The structure of the bridge does not determine the sign of $G'(1)$. Both
signs occur, and each is strict in the examples below, so each persists on
a relatively open set within the gateway-model parameterization.

\subsection{Identical clusters}

Give each of $k$ clusters a knot of unit variance and one independent asset
of variance $1/\delta$, and let every pair of knots have correlation $c$.
By \cref{cor:compress} this is less special than it looks. By symmetry a
portfolio is described by one number, the total weight $x$ on the knots,
and with $L=1+(k-1)c$ its variance is
\begin{equation}
  V(x)=\frac1k\Big[L\,x^2+\frac{(1-x)^2}{\delta}\Big],
  \qquad
  V(x)-V(x_\star)=\frac{L+1/\delta}{k}\,(x-x_\star)^2,
  \qquad x_\star=\frac{1}{1+\delta L}.
  \label{eq:scalar}
\end{equation}
The bridge is estimating the split between a correlated basket and an
independent one.

If the estimated knot correlation is $z$, \cref{prop:knotsmove} gives the
member part $\delta$ and the knot exposure
\begin{equation}
  t_\gamma(z)=\frac{1+(k-2)z-\gamma(k-1)z}{1+(k-2)z-\gamma(k-1)z^2},
  \qquad\text{so}\qquad
  x_\gamma(z)=\frac{t_\gamma(z)}{t_\gamma(z)+\delta}.
  \label{eq:tgamma}
\end{equation}
No matrix inversion remains. For $0<z<1$ the exposure is strictly
decreasing in $\gamma$, because $\partial_\gamma t_\gamma(z)$ has the sign
of $-z(k-1)(1-z)\big(1+(k-2)z\big)$.

\paragraph{An exact interior optimum.} Let the estimate be $Z\in\{0,2c\}$
with equal probability, for some $0<c<\tfrac12$. On the branch $Z=0$
damping does nothing. On the branch $Z=2c$ the choice
\begin{equation}
  \gamma_\star=\frac{1+2(k-2)c}{2\,\big[1+(k-3)c\big]}
  \label{eq:gammastar}
\end{equation}
gives $t_{\gamma_\star}(2c)=1/L$ and hence $x=x_\star$, the population optimum.
One branch is constant and the other attains its absolute minimum, uniquely
since $t_\gamma$ is strictly monotone. So $\gamma_\star$ is the unique
global optimum. It lies strictly inside $(0,1)$ and does not depend on
$\delta$.

With two clusters, $c=\tfrac14$ and $\delta=1$ the optimum is
$\gamma_\star=\tfrac23$, and in exact arithmetic $F(0)-F(\tfrac23)=1/576$
and $F(1)-F(\tfrac23)=1/900$.

\paragraph{Both signs in one model.} Now let the noise be small and
symmetric, $Z=c\pm\tau$, with $k=10$ and $c=\tfrac14$. Then \eqref{eq:shift}
becomes
\begin{equation}
  \tilde\gamma(\tau)=1-\frac{8-\delta}{1+13\delta/4}\,\tau^2+O(\tau^4).
  \label{eq:threshold}
\end{equation}
For $\delta<8$ the optimum is interior for small noise, for $\delta>8$ it
stays at full coupling, and at $\delta=8$ higher orders decide. Here $V_0$
has a unique minimizer, again because $t_\gamma$ is monotone, and
$F(\cdot,\tau)\to V_0$ uniformly, so by \cref{prop:mvend} these conclusions
are global on the bridge.

Take $\delta=12$. Then $\tilde\gamma=1+\tau^2/10+O(\tau^4)$ and the
constrained optimum stays at $1$ under small noise. At the larger noise
level $\tau=\tfrac14$, which is $Z\in\{0,\tfrac12\}$, \eqref{eq:gammastar}
gives $\gamma_\star=10/11$. The same symmetric model sits at full coupling
under small noise and moves inside under larger noise.

\paragraph{Why the sign can go either way.} At full coupling
$x_1(Z)=1/\big(1+\delta[1+(k-1)Z]\big)$, which is convex in $Z$. Unbiased
correlation noise therefore biases the expected knot allocation upward.
Damping toward NCO also raises the knot allocation when correlations are
positive. It reduces sensitivity to the noise and can worsen that bias, and
the sign of $G'(1)$ is the outcome of the competition.

\subsection{Less symmetric examples}

\paragraph{A quartic.} Take clusters $\{1,2\}$ and $\{3,4\}$ with knots $1$
and $3$, unit variances except $\Sigma_{33}=4$, a single nonzero
correlation entry $\Sigma_{13}=\tfrac12$, and the estimate
$\widehat\Sigma_{13}=Z\in\{0,1\}$ with equal probability. On the branch
$Z=1$ the second knot exposure is $x=\kappa_2(\gamma)=(1-\gamma)/(4-\gamma)$
and the population variance is
$V(x)=(49x^4-7x^3+36x^2-2x+6)/\big(2(7x^2+3)^2\big)$. Its derivative has the
sign of $P(x)=49x^4+84x^3-21x^2+48x-6$, which is increasing on
$[0,\tfrac14]$ with $P(0)=-6$ and $P(\tfrac14)=1585/256$.

So there is a unique optimum $\gamma_\star\approx0.5587$, strictly inside
the bridge, with $F(0)-F(\tfrac12)=3/1210$ and $F(1)-F(\tfrac12)=5/1452$ in
exact arithmetic. With small symmetric noise $Z=\tfrac12\pm\tau$ the same
covariance has $G'(1)=31779128936/554523204167$ and
$V_0''(1)=48136/3571279$, and \eqref{eq:shift} predicts
$\tilde\gamma(0.05)\approx0.9894$ against a computed optimum of $0.9896$.

\paragraph{Full coupling under asymmetric noise.} Take three clusters, each
a knot and one remaining member that is uncorrelated with everything, with
member variances $1$, $4$ and $8$. Let the knot covariance and its estimates
be
\begin{equation*}
  S=\begin{pmatrix}5&6&5\\6&\tfrac{33}{2}&12\\5&12&\tfrac{25}{2}\end{pmatrix},
  \qquad
  \widehat S_\tau=S\pm\tau\begin{pmatrix}0&-1&0\\-1&0&2\\0&2&0\end{pmatrix},
\end{equation*}
the two signs equally likely. The noise is symmetric, touches only
cross-knot entries, and preserves the gateway model. Exact differentiation
gives $V_0''(1)\approx2.3949$ and
\begin{equation*}
  G'(1)=-\frac{4720238484060245648386156800}{806907939268294475102923640251}
  \approx-0.0058498<0,
\end{equation*}
so by \cref{prop:mvend} full coupling is a strict local minimizer for all
sufficiently small $\tau$. As a spot check, at $\tau=\tfrac1{10}$ and
$\tau=\tfrac1{100}$ exact evaluation gives $F(\gamma,\tau)>F(1,\tau)$ for
$\gamma\in\{0.999,\,0.99,\,0.9,\,0.5,\,0\}$, which is consistent with a
global optimum there but does not prove one.

Neither end owns the bridge. An interior optimum is not a consequence of the
architecture, and a statement that it is generic needs a noise model. The
derivatives quoted in both examples are exact, computed by truncated
bivariate jets in rational arithmetic.

\section{Rank deficiency}
\label{sec:rank}

Everything so far assumes $\Sigma\succ0$. This section allows
$\Sigma\succeq0$.

A singular covariance is one of two kinds. If some null portfolio has a
nonzero total, a fully invested portfolio of zero variance exists. If every
null portfolio is self-financing, the redundant assets can be removed, and
only then is the pseudoinverse formula
$w_\star=\Sigma^\dagger\ones/\ones^\top\Sigma^\dagger\ones$ valid. The
condition is $\ones\in\operatorname{range}(\Sigma)$. For
$\Sigma=\big(\begin{smallmatrix}1&2\\2&4\end{smallmatrix}\big)$ the
portfolio $(2,-1)$ has budget $1$ and variance $0$, while normalizing
$\Sigma^\dagger\ones$ gives $(\tfrac13,\tfrac23)$ with variance $25/9$.

\paragraph{Harmless redundancy breaks the recipe.} Take two clusters, each a
knot and an independent asset, all of unit variance, and make the two knots
duplicates. The budget vector lies in the range of $\Sigma$, the optimum is
$w_\star=(\tfrac16,\tfrac13,\tfrac16,\tfrac13)$ in the order knot, member,
knot, member, and $V_\star=\tfrac13$. For every $\gamma<1$ the pairs are
$Q_i=\operatorname{diag}(1-\gamma,1)$ and $b_i=(1-\gamma,1)$, so both
directions are $(1,1)$ and the bridge holds equal weights.

\begin{table}[ht]
\centering
\begin{tabular}{@{}lr@{}}
\toprule
Construction & Variance \\
\midrule
Global optimum & $1/3$ \\
Bridge at every $\gamma<1$ & $3/8$ \\
Bridge at $\gamma=1$ with each $Q_i^{-1}$ replaced by $Q_i^\dagger$ & $1/2$ \\
\bottomrule
\end{tabular}
\caption{Duplicated knots.}
\label{tab:duplicate}
\end{table}

At full conditioning each knot is completely explained by the other, and
each local pseudoinverse discards its own knot. Together they discard a
common return that is still useful. The singular Schur equations leave some
weights undetermined, and choosing a minimum-norm solution separately in
each cluster does not coordinate those choices.

\Cref{prop:effective} explains the approach to this case. As a knot becomes
redundant $r_i\to0$, so at every fixed $\gamma<1$ its effective damping
tends to zero and all the movement is compressed into a shrinking interval
near $\gamma=1$. The limits do not commute. With the knot covariance
regularized to $\ones\ones^\top+\varepsilon I$,
\begin{equation*}
  \lim_{\varepsilon\downarrow0}\,\lim_{\gamma\uparrow1}V_0=\tfrac13,
  \qquad
  \lim_{\gamma\uparrow1}\,\lim_{\varepsilon\downarrow0}V_0=\tfrac38 .
\end{equation*}
This is more than numerical instability. The parameterization by $\gamma$
is itself singular there.

\begin{proposition}[A bridge in effective exposure]
\label{prop:lambda}
Let the knot covariance $S\succeq0$ have positive diagonal, let
$\ones\in\operatorname{range}(S)$, and let every $\delta_i>0$. Put
$h=S^\dagger\ones$ and, for $\lambda\in[0,1]$,
\begin{equation}
  \bar\kappa_i(\lambda)=(1-\lambda)\,\frac1{s_i}+\lambda\,h_i ,
  \label{eq:lambdapath}
\end{equation}
keep the member parts of \cref{prop:knotsmove}, and allocate across the
resulting cluster portfolios as before. The outer matrix
$\bar K S\bar K+\Delta$ is positive definite, the portfolio is continuous in
$\lambda$, it is NCO at $\lambda=0$, and it is minimum variance at
$\lambda=1$.
\end{proposition}

\begin{proof}
Positive definiteness follows from $\Delta\succ0$. At $\lambda=1$,
$Sh=\ones$ gives
$\big(\operatorname{diag}(h)\,S\,\operatorname{diag}(h)+\Delta\big)\ones=h+\delta$,
so the outer solution is $a=\ones$ and the portfolio has knot exposures
proportional to $h$ and member weights proportional to $E_i^{-1}(\ones-\beta_i)$,
which is the minimum-variance portfolio of the compressed problem.
\end{proof}

For duplicated knots $\bar\kappa_i=1-\lambda/2$, and the portfolio moves
continuously from variance $3/8$ to $1/3$. When $S\succ0$ this path gives
every cluster the same effective damping $\lambda$, which corresponds to
cluster-specific parameters $\gamma_i=\lambda/\big(r_i+(1-r_i)\lambda\big)$.
It is a different interior path with the same two ends, and we do not claim
it performs better out of sample. It makes a design choice visible: whether
clusters should share a nominal $\gamma$ or an effective $\lambda$.

\begin{remark}[Noise around a singular covariance]
If $\Sigma$ is singular and both $\Sigma+\tau E$ and $\Sigma-\tau E$ are
positive semidefinite for some $\tau>0$, then $E\ker(\Sigma)=0$. For
$v\in\ker\Sigma$ both signs force $v^\top Ev=0$, and a positive semidefinite
matrix with $v^\top(\Sigma+\tau E)v=0$ has $(\Sigma+\tau E)v=0$. Symmetric
additive noise must therefore preserve the structural null space. A singular
sample estimate of a positive definite population covariance is a different
problem, since its empirical null directions are not riskless. In
\cref{prop:mvend} the condition $\dot w\neq0$ becomes
$\dot w^\top\Sigma\dot w>0$.
\end{remark}

\section{Cost}

No step forms an $n\times n$ inverse. The largest linear solve has dimension
$\max(k,\max_i|I_i|)$, as in NCO: the outer tier is $k\times k$ and each
conditioning inverse is $(k-1)\times(k-1)$. The bridge adds work that NCO
does not do: the cross-block products $\Sigma_{i,P_{-i}}$, and a
leave-one-knot-out inverse per cluster, obtainable from $\Sigma_{PP}^{-1}$
by a rank-one downdate.

If the gateway model fails, \eqref{eq:cheap} approximates the full
complement in \eqref{eq:complement}, and $\Delta_i$ in \eqref{eq:loss} is
what is lost.

\section{Closing}

Nested clustered optimization has been treated as belonging to a different
family from optimization proper: a clustering heuristic that quarantines the
instability of Markowitz rather than a relaxation of it. This paper shows
otherwise. Under a rank-one model of cross-cluster dependence, a Schur
relaxation places NCO and the global minimum-variance portfolio at the two
ends of a one-parameter family of portfolios, at a cost no greater than
NCO's own.

The family has structure. Along it only the knots move, each by an explicit
fraction of its journey, and the variance given up at every setting is in
closed form. It also comes with conditions for where to sit under
estimation error. The sign of one derivative at the minimum-variance end
decides whether remaining at that endpoint is optimal out of sample, and a
symmetric family exhibits both outcomes in a single model, with the optimal
coupling exact.

The separation reasserts itself when the knot covariance is rank deficient.
There the parameterization by $\gamma$ is itself singular, and NCO and the
optimum stay apart on it, though a path in effective exposure still joins
the two ends. The ingredients throughout are old: block inversion, the
truncation of a conditioning set as in Vecchia approximations, and the
damping of \cite{cotton2024schur}.

Two directions are open. \Cref{prop:sufficient} may serve
Schur-complementary allocation on the tree, where the complement against a
sibling block could be taken against that block's knots. The whole-matrix
analysis of \cite{cotton2026bridge} rests on a scalar reduction specific to
the tree and needs a replacement here. We leave both, and an empirical
comparison, to future work.

\end{document}